\documentclass[conference]{IEEEtran}
\IEEEoverridecommandlockouts
\usepackage{cite}
\usepackage{amsmath,amssymb,amsfonts}
\usepackage{amsthm}
\newtheorem{theorem}{Theorem}[section]
\usepackage{booktabs}
\usepackage{xurl}   

\usepackage{algorithm}
\usepackage{algpseudocode}
\usepackage{graphicx}
\usepackage{textcomp}
\usepackage{xcolor}
\usepackage{tikz}
\usepackage{tikz-cd}
\usepackage{algorithm}
\usepackage{booktabs}
\usepackage{tabularx}
\usepackage{array}

\newcolumntype{Y}{>{\raggedright\arraybackslash}X}

\def\BibTeX{{\rm B\kern-.05em{\sc i\kern-.025em b}\kern-.08em
    T\kern-.1667em\lower.7ex\hbox{E}\kern-.125emX}}
\begin{document}

\title{Secure AltDA Integration for Ethereum L2s: An End-to-End Validation Framework\\
}

\author{

\IEEEauthorblockN{1\textsuperscript{st} Bowen Xue}
\IEEEauthorblockA{
Eigenlabs, Inc \\
Seattle, WA }
\and
\IEEEauthorblockN{2\textsuperscript{st} Samuel Laferriere}
\IEEEauthorblockA{
Independent \\
New York, NY \\}
}

\maketitle

\begin{abstract}
Alternative data availability (AltDA) systems provide Ethereum L2s with an external data publication layer for high throughput rollup designs. By moving bulk data publication outside of Ethereum, AltDA allows L2s to process more data than native DA. However, this replacement introduces a new consensus critical integration layer. Existing ecosystem frameworks identify high level risks, such as external DA trust assumptions and the presence or absence of a DA verifier, but do not provide a complete specification for how an L2 should integrate with AltDA. This gap can lead to L2 halts, inconsistent derivation across honest L2 nodes, invalid state assertions, or bridge attacks. This paper presents a canonical validation framework for secure AltDA integration. We model the boundary as a typed, deterministic, and total translation from L1 inbox bytes to an AltDA commitment, then to externally available data, and finally to the rollup payload consumed by the rest of core L2s logic. The central principle is that every adversarial input must lead to a defined unique outcome. We show how missing obligations lead to concrete failure modes, including underconstrained settlement, derivation halts, inconsistent honest node behavior, invalid state assertions, and bridge safety failures. We then apply the framework to representative AltDA integration architectures, including Celestia-Blobstream, EigenDA based designs, and Avail-ZKsync. Our evaluation shows that secure AltDA integration is not determined solely by the DA provider or bridge. The surrounding L2 integration must also enforce the full validation relation connecting L1 inbox inputs to accepted L2 state.
\end{abstract}


\section{Introduction}

Ethereum scales execution through Layer 2 (L2) systems. An L2 is itself a
blockchain, but its state machine is executed outside Ethereum consensus rather
than by Ethereum validators as part of L1 block validation. Ethereum instead
acts as a settlement layer: it orders rollup inputs and enforces a mechanism for
accepting or rejecting claimed L2 state roots. Depending on the rollup design,
this mechanism may be optimistic, validity based, or a hybrid of the two. This
separation between L2 execution and L1 settlement lets rollups amortize L1 costs
across many transactions and specialize their execution environments, making L2
execution cheaper than executing the same workload directly on Ethereum L1.

This separation only works when the data needed to reconstruct the L2 state
remains available. Honest participants must be able to derive the L2 chain,
challenge invalid state assertions, or generate validity proofs. Native
Ethereum data availability provides this property by requiring rollup data to be
published through Ethereum calldata or blobs.
However, Ethereum's native DA capacity is limited. 
After the Fusaka upgrade\cite{fusaka_upgrade}, Ethereum supports up to
21 blobs per 12 second slot. Since each blob is 128 KiB, this corresponds to
roughly 0.219 MiB/s of blob data. These limits define the operating boundary for
standard Ethereum DA rollups.
This capacity is large enough for many rollups, but it becomes a bottleneck as L2 execution engines target much higher throughput, including designs that target GigaGas-per-second execution~\cite{megaeth_gas,rise_gas,conduit_g3_gas}.

Alternative data availability systems address this bottleneck by moving bulk rollup data outside Ethereum. Instead of posting the entire batch to Ethereum, an AltDA backed rollup publishes data to an external DA layer, such as Celestia\cite{celestia}, EigenDA\cite{eigenda}, Avail\cite{avail_DA}, or a committee based DA system, and posts a compact reference to Ethereum.
This design is no longer merely experimental. According to ecosystem dashboards, AltDA backed L2s already secure more than 800 millions value~\cite{l2_beat_tvl} at the time of writing. As a result, the correctness of AltDA integration has become a direct security concern for L2s and bridges.

Moving data outside Ethereum changes the L2's security boundary. In a native DA rollup, the data consumed by derivation is the data ordered and made available by Ethereum. In an AltDA L2, Ethereum orders an reference object, while the full data lives in an external DA system. The L2 integration must therefore connect three domains: the object visible on L1, the external data authenticated by the DA system, and the rollup payload or execution trace that determines the state accepted by settlement.

This connection is easy to underspecify. A DA verifier\cite{l2beat_da_verifier} may authenticate a statement about external data without proving that
the L2 state is derived from the same data. A commitment may have been valid when it was created but unsafe to consume after the DA system has pruned the data. An available blob may be correctly bound to a DA commitment while still failing to decode into a unique rollup payload. 
These are not merely software engineering edge cases. They are knowledge gaps that affect whether honest nodes agree on derivation and whether the settlement contract accepts only state transitions constrained by the authenticated DA data.

Existing ecosystem classifications~\cite{l2beat_da_verifier} identify important
high level risks, such as whether an L2 uses external DA, whether a DA verifier (bridge)
exists, and what trust assumptions the DA provider introduces.
A prior online
discussion also raised informal concerns about secure AltDA
integration~\cite{ethresearch-altda-framework}.
These discussions are valuable,
but they do not fully specify the end-to-end relation that an AltDA L2 must enforce.
A DA verifier answers a DA provider level question: whether a fact about an external DA system can be authenticated on Ethereum. 
Secure rollup integration requires a different question: whether every adversarial input at the L2's DA boundary leads to a deterministic and settlement enforced result.

This paper argues that secure AltDA integration should be understood
an end-to-end validation problem.
The key point is that the validation boundary starts with arbitrary bytes
ordered by the rollup's L1 inbox, not with an already well formed DA commitment.
A secure AltDA integration must deterministically classify each such input as a
unique rollup payload, a rejection, or a specified stall condition. These
obligations may be enforced in derivation, in a fault proof VM, in a validity
proof, or at the inbox itself.
%
This paper makes the following contributions:
\begin{itemize}
    \item We introduce a typed validation model for AltDA L2s that
    separates L1 inbox bytes, AltDA commitments, external blobs, and rollup
    payloads.

    \item We identify the core semantic obligations required for secure AltDA
    integration: parsing, DA verification, recency, commitment-to-blob binding, payload
    uniqueness, and deterministic failure handling. 

    \item We show how these obligations can be enforced at different points in
    the rollup settlement path, including interactive fault proofs, validity
    proofs, and eager inbox verification.
    
    \item We derive an attack taxonomy showing how missing validation
    obligations can lead to underconstrained settlement, derivation halts,
    honest node divergence, or unchallengeable state transitions.

    \item We evaluate representative AltDA integrations through this framework
    and identify which obligations are enforced by the DA bridge, the inbox,
    the derivation pipeline, or the proof system.
    
\end{itemize}

This paper does not study DA layer itself. It studies the security critical
interface between external DA systems and L2 settlement. Our main claim is
that secure AltDA integration is an end-to-end property of the inbox,
derivation, proof, and settlement path, not merely a property of the DA provider
or verifier.

\section{Background and Threat Model}
\label{sec:background}

\subsection{Rollup Derivation and Settlement}
\label{sec:background:derivation}

A rollup~\cite{op_spec,arb_doc} maintains an L2 chain whose
canonical history is derived from data ordered by an underlying L1 chain. Unlike
an independent blockchain, where validators directly determine the canonical
chain, a rollup's L2 history is reconstructed by applying a deterministic
derivation function to L1 ordered inputs.

The main L1 input source is the rollup inbox. Informally, the inbox is the L1
contract interface or data location through which rollup inputs become available
to L2 nodes. It may include sequencer batches, deposits, forced transactions,
upgrade messages, and other system inputs that affect L2 derivation. Honest L2
nodes monitor the canonical L1 chain, identify rollup relevant inbox data, and
apply the rollup's derivation rules.

The sequencer collects L2 transactions, orders them, and submits batches to the
rollup's L1 input source. Batch data may be included directly as calldata or
carried as Ethereum blob data~\cite{eip_4844}. The sequencer proposes an
ordering, but it does not unilaterally define the canonical L2 chain. A batch
becomes part of the L2 history recognized by honest nodes only when it is posted
to the appropriate L1 input source and accepted by the rollup's derivation
rules.

Rollups ultimately use L1 to enforce or finalize claims about L2 state. In an
optimistic rollup, a proposer posts an asserted L2 state root to L1, and the
assertion is accepted after a challenge window unless successfully disputed. In
a validity rollup, the L1 settlement contract accepts state updates only with a
validity proof showing that the transition follows the rollup's rules. Hybrid
designs may combine these mechanisms. In all cases, the settlement protocol
relies on a well defined relation between L1 ordered input data and the accepted
L2 state.

We use L2 as the default term for the systems studied in this paper, including
both native DA and AltDA integrations. We occasionally use rollup terminology when we specifically discuss the integration with native DA, or when referring to standard components shared across nativeDA and AltDA, such as the
rollup inbox, rollup payload.

\subsection{AltDA Commitments and DA Verifier}

An alternative data availability system allows an L2 to publish bulk data
outside Ethereum's native DA layer while posting only a compact reference or
commitment to Ethereum. We use \(\mathsf{AltDACommitment}\) as the general term
for the structured object consumed from or derived from the L1 inbox. Depending
on the DA system, an \(\mathsf{AltDACommitment}\) may be a certificate containing an aggregated signature, an inclusion proof
reference, a committee attestation, a bridge verified commitment, or a pointer
to data under an authenticated DA root.
The important difference between AltDACommitment
from
native DA is that Ethereum orders it, but Ethereum consensus does not
by itself verify how recent and if the AltDACommitment corresponds to some blob that indeed has been stored by the DA system.
A DA Verifier is a smart contract deployed in Ethereum, that has abilities to verify if certain AltDACommitment has been stored and attested by the AltDA network.

\subsection{Threat Model}

The adversary may control the sequencer, batcher, or any actor capable of influencing bytes posted to the L1 inbox. It may post arbitrary byte strings, including malformed commitments, invalid attestations, stale commitments, incorrect metadata, or commitments to data that cannot be interpreted as valid rollup input. The adversary may also operate malicious or unreliable blob serving infrastructure. It may refuse to serve data, serve incorrect data, serve data that is correctly bound but semantically invalid for the rollup, or provide different responses to different honest nodes.

The adversary may also submit malicious L2 state claims to settlement. We assume the cryptographic assumptions of the underlying DA system hold. In particular, we do not consider attacks that break the binding of the commitment scheme\cite{kzg} or the security of the DA attestation mechanism. Our focus is the integration boundary between the DA system and the L2.

\section{Canonical Model for AltDA Integration}
\label{sec:model}
    \subsection{Why AltDA Requires Additional Validation}

AltDA differs from native Ethereum DA in two ways: validity and atomicity.

With Ethereum native DA, data availability validity is enforced by Ethereum
consensus. If a rollup batch is published through calldata or an Ethereum blob
transaction, the data needed by the rollup is part of Ethereum's native DA model. A
block containing a blob transaction is valid only if the corresponding blob is
available according to Ethereum consensus\cite{eip_4844}. Thus, an L2 consuming native
Ethereum DA inherits the DA validity of the input data from Ethereum itself.

With AltDA, Ethereum orders only an object that refers to external data. Ethereum consensus does not by itself establish that the external data was published, attested, retained, or retrievable under the rules of the external DA system. The integration must introduce a verification mechanism, which may be implemented by an L1 inbox contract, a DA verifier contract, L2 derivation software, a fault proof VM, a validity circuit, or a combination of these components.

The second difference is atomicity. In Ethereum native DA, ordering and data
availability are tied together by Ethereum block validity. In AltDA, this
atomicity is not automatic. A sequencer may obtain or construct an AltDA
commitment at one time and post it to the L1 inbox at another time. A commitment
that was once valid may no longer be safe to consume later, especially if the DA
system has pruned the data or if honest parties no longer have sufficient time
to retrieve the data and participate in the L2's dispute or proof process. The
integration must therefore define a recency rule for commitments consumed by
the L2.

These differences motivate two core obligations. DA verification restores the missing validity relation between the L1 visible object and the external DA system. Recency restores the missing timing relation between L1 ordering and external availability.

\subsection{Typed Boundary}

We model the AltDA boundary as a sequence of typed objects:
\[
\resizebox{0.45\textwidth}{!}{$
\mathsf{InboxBytes}
\longrightarrow
\mathsf{AltDACommitment}
\longrightarrow
\mathsf{Blob}
\longrightarrow
\mathsf{RollupPayload}.
$}
\]

Although written as four typed objects, this pipeline connects three
semantic domains. The first domain is the L1 inbox domain: \(InboxBytes\)
are ordered by Ethereum as long as it is a valid eth transaction, 
but the transaction can be adversarial in content. The second domain is the
AltDA domain: \(AltDACommitment\) and \(Blob\) are interpreted under the
semantics of the external DA system. The third domain is the rollup derivation
domain: \(RollupPayload\) is the input consumed by the ordinary L2 derivation
and execution logic.

Thus, the core AltDA semantics live in the middle domain, while the two outer
transitions glue that domain to the surrounding L2 system. This distinction
is important because validity in one domain does not imply validity in another:
a valid AltDA commitment does not by itself determine a valid rollup payload,
and a payload consumed by execution is secure only if settlement enforces its
relation to the accepted AltDA commitment.

\subsection{AltDA Semantics: Verification, Recency, and Binding}
The AltDA specific part of the validation relation can be written as:

\[
\begin{tikzcd}
\mathsf{AltDACommitment} \arrow[d, "\mathsf{verify/validDA}"] \\
\mathsf{VerifiedCommitment} \arrow[d, "\mathsf{recent}"] \\
\mathsf{RecentVerifiedCommitment} \arrow[d, "\mathsf{retrieve/bind}"] \\
\mathsf{Blob}.
\end{tikzcd}
\]

The \(validDA\) predicate establishes that the commitment is valid under the external DA system. This may mean that a quorum signed a certificate, that a data root was finalized and bridged to Ethereum, or that an inclusion proof verifies under an authenticated DA root.

A \(recent\) predicate is defined relative to both the external DA system
and the rollup settlement window. Let \(t_{\mathsf{DA}}(c)\) denote the DA layer
publication or attestation time of commitment \(c\), \(t_{\mathsf{L1}}(c)\)
denote the L1 time at which \(c\) is posted on L1, and
\(W_{\mathsf{settle}}\) denote the maximum time honest parties may need to
retrieve the blob and use it in derivation, proof generation, or dispute. If
the DA layer retains data for \(W_{\mathsf{DA}}\), then a typical recency rule
requires
\[
    t_{\mathsf{L1}}(c) - t_{\mathsf{DA}}(c) + W_{\mathsf{settle}}
    \leq W_{\mathsf{DA}} .
\]
The exact clocks and margins are integration specific, but the purpose is
general: a commitment should be accepted only while the corresponding data is
still expected to be retrievable by honest parties.

The \(bind\) predicate establishes that the retrieved blob is the blob committed by the accepted AltDA commitment. 
For a commitment \(c\) and blob \(b\), the integration must establish
\[
    \mathsf{bind}(c,b)=1
\]
Formally,
for any commitment \(c\), there is only negligible probability $\varepsilon(\kappa)$ for all adversary to find distinct
blobs \(b \neq b'\) that commits to \(c\), where $\varepsilon(\kappa)$ measure bits of security
\begin{equation}
\label{eq:binding}
\Pr\left[
    \mathsf{bind}(c,b)=1
    \land
    \mathsf{bind}(c,b')=1
    \right]
\leq \varepsilon(\kappa),
\end{equation}

This assumption may be instantiated by KZG commitments~\cite{kzg}, Merkle commitments, or other authenticated data structures.

In summary, those predicates cover the additional validation required by the AltDA, and they ensure only one blob can be provided for the downstream L2 derivation, which ensures the uniqueness.

\subsection{Total Validation}

The integration must ensure every possible inputs are met with a well defined, deterministic output.
We model derivation based validation as a total function
\[
    F :
    \mathsf{InboxBytes}
    \rightarrow
    \mathsf{RollupPayload}
    \cup
    \{\mathsf{Drop}, \mathsf{Stall}\}.
\]
For an inbox byte string \(x\), \(F(x)=\mathsf{RollupPayload}(p)\) means that
\(x\) deterministically yields the rollup payload \(p\). \(F(x)=\mathsf{Drop}\)
means that \(x\) is invalid and must be skipped. \(F(x)=\mathsf{Stall}\) means
that the input is not yet safe to consume, for example because it refers to a
valid and recent commitment whose blob is temporarily unavailable due to RPC infrastructure issue or outage of the AltDA system.

The distinction between \(Drop\) and \(Stall\) is important. Invalid data should not halt derivation forever. It should be rejected deterministically.
The \(Stall\) outcome is a bounded liveness condition, not an unbounded
consensus state. Implementations should pair it with retry and timeout rules so
that temporary unavailability is handled according to the L2's specified
liveness policy.

\begin{algorithm}[t]
\caption{AltDA validation semantics}
\label{alg:validation}
\begin{algorithmic}[1]
\Require \(x \in \mathsf{InboxBytes}\)
\Ensure one of \(\mathsf{RollupPayload}(p)\), \(\mathsf{Drop}\), or \(\mathsf{Stall}\)

\State \(c \gets \mathsf{parse}(x)\)
\If{\(c = \bot\)}
    \State \Return \(\mathsf{Drop}\)
\EndIf

\If{\(\neg \mathsf{valid}_{\mathsf{DA}}(c, s)\)}
    \State \Return \(\mathsf{Drop}\)
\EndIf

\If{\(\neg \mathsf{recent}_{\mathsf{L2}}(c)\)}
    \State \Return \(\mathsf{Drop}\)
\EndIf

\State \(b \gets \mathsf{retrieve}(c)\)
\If{\(b = \mathsf{Unavailable}\)}
    \State \Return \(\mathsf{Stall}\)
\EndIf

\If{\(\neg \mathsf{bind}(c,b)\)}
    \State \Return \(\mathsf{Drop}\)
\EndIf

\If{there is no unique \(p\) such that \(\mathsf{payloadRel}(b,p)\)}
    \State \Return \(\mathsf{Drop}\)
\EndIf

\State \Return \(\mathsf{RollupPayload}(p)\)
\end{algorithmic}
\end{algorithm}

Algorithm~\ref{alg:validation} presents the validation semantics. The relation \(payloadRel(b,p)\) captures the L2 specific codec: the blob \(b\) must decode, under the L2's canonical rules, to exactly one payload \(p\). This condition rules out non-canonical encodings, ambiguous padding, implementation dependent parsing, or payload formats that can cause honest nodes to diverge.

\begin{theorem}[Total validation and uniqueness]
For every inbox byte string \(x \in \mathsf{InboxBytes}\),
Algorithm~\ref{alg:validation} returns exactly one of
\(\mathsf{RollupPayload}(p)\), \(\mathsf{Drop}\), or \(\mathsf{Stall}\).
If it returns \(RollupPayload(p)\) for commitment \(c\), then, except with negligible probability under the binding assumption, the accepted commitment \(c\) determines a unique rollup payload \(p\).
\end{theorem}

\begin{proof}
Algorithm~\ref{alg:validation} evaluates a finite sequence of validation cases.
Parsing failure, DA verification failure, stale commitment, binding failure,
and non-unique or failed payload interpretation return \(\mathsf{Drop}\).
Temporary unavailability of the blob for an otherwise valid and recent
commitment returns \(\mathsf{Stall}\). If all validation checks succeed, the
algorithm returns the unique \(\mathsf{RollupPayload}(p)\) selected by
\(\mathsf{payloadRel}\). These cases cover every branch of the algorithm, and
each branch returns exactly one output.

For uniqueness, suppose the algorithm returns \(\mathsf{RollupPayload}(p)\)
for a commitment \(c\). The algorithm only reaches this branch after retrieving
a blob \(b\), checking \(\mathsf{bind}(c,b)\), and checking that exactly one
payload \(p\) satisfies \(\mathsf{payloadRel}(b,p)\). By the binding assumption,
no adversary can produce a distinct blob \(b'\) that also binds to \(c\), except
with negligible probability. By functionality of \(\mathsf{payloadRel}\), the
accepted blob determines at most one rollup payload. Therefore the accepted
commitment uniquely determines the returned rollup payload, except with
negligible probability.
\end{proof}

The theorem is intentionally simple: it states the local correctness property that the validation relation must provide. The security significance comes from enforcement of all the constraints.
If settlement omits any part of this relation, a valid DA commitment may become disconnected from the state transition accepted on L1.
Section   \ref{sec:integration_with_L2_settlement} reveals a key challenge for implementing \(validDA\) predicate.

\subsection{Semantic Obligations, Not Required Modules}

The validation steps in Algorithm~\ref{alg:validation} are semantic obligations, not necessarily separate software modules. An implementation may combine verification and recency, or enforce verification and binding inside a fault proof step. A validity proof may prove the DA predicate and payload transformation in one circuit. What matters is not where the checks appear, but whether the settlement path enforces the same end-to-end relation.

\section{Integration with L2 Settlement}
\label{sec:integration_with_L2_settlement}

Algorithm~\ref{alg:validation} describes the local validation relation. This section explains where that relation can be enforced in an end-to-end L2 system.


%

At a high level, an L2 is designed so that Ethereum maintains an authenticated view
of the current L2 state.
This is typically implemented by an L1
settlement contract to which proposers periodically submit claimed L2 state
roots. 
%
%
The contract does not reexecute the full L2 state transition function
in the normal case. Instead, it acts as an \emph{L2 light client}: a
lightweight verifier that tracks accepted L2 state and rejects incorrect state
claims through a fault proof, validity proof, or hybrid mechanism; see
Figure~\ref{fig:l1_l2_boundary}.

For an AltDA L2, this settlement
path must enforce not only execution correctness, but also the DA relation used
to derive the executed payload, as shown in Algorithm~\ref{alg:validation}. Otherwise, L1 may accept a state transition that
is not constrained by the AltDA commitment posted to the inbox.

The main challenge comes from the \(validDA\) predicate. Validating
an \(AltDACommitment\) often requires L1 state, such as DA verifier
contract on Ethereum, while
the validation itself needs to be a part of L2 consensus. The rest of this
section describes three patterns for connecting relevant L1 state information
to L2 DA verification: interactive fault proof integration,
validity proof integration, and eager inbox verification.

\begin{figure}[h]
    \centering
    \includegraphics[width=0.48\textwidth]{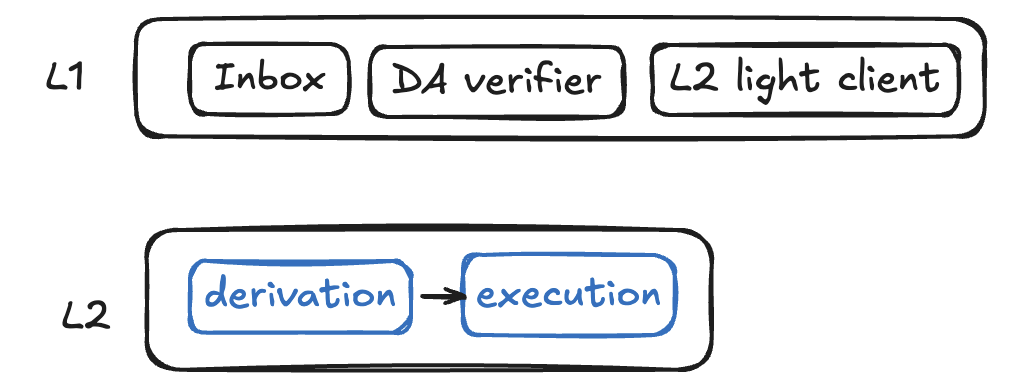}
    \caption{Separation between L1 and L2 for AltDA integration, the limitation of getting L1 state to L2 constitutes the largest challenge for most AltDA integrations.}
    \label{fig:l1_l2_boundary}
\end{figure}

\subsection{Interactive Fault Proof Integration}

In an optimistic rollup with interactive fault proofs, a proposer posts a claimed L2 state root. A challenger can dispute the claim, and the parties bisect the execution trace until the dispute reduces to a single machine step that can be resolved onchain.

In this model, AltDA validation can be carried as a part of the disputed state transition relation.
If the derivation of the challenged L2 state depends on an AltDACommitment, then
the fault proof system must be able to determine whether that commitment is valid.
One natural design is to place the AltDA verifier on L1 and allow the L2 light
client to call it when the disputed machine step reaches the AltDA validation
logic.

\subsection{Validity Proof Integration}
\label{sec:e2d_validity_proof}

In a validity L2, the settlement contract does not replay a disputed execution step. Instead, it verifies a succinct proof that the prover correctly performed the relevant derivation and execution.

For AltDA, the validity proof must include the DA relation in its statement. It is not sufficient to prove that some rollup payload leads to a claimed state transition. The proof must also establish that the payload was derived from data authenticated by the accepted AltDA commitment as shown in Algorithm\ref{alg:validation}. Because the DA verifier depends on Ethereum state, the prover often have to use storage proofs on ethereum merkle patricia trie\cite{merkle_patricia_trie} to show that the relevant verifier state existed at a particular L1 block. The proof system can then use the proven L1 state to check the DA predicate against those authenticated inputs.

\subsection{Eager DA Verification}

A third pattern is to verify AltDA commitments before they enter rollup derivation. This is most natural when the rollup inbox is a smart contract. The inbox contract can call the DA verifier at submission time and record only verified inbox entries. L2 derivation then consumes pre-validated objects rather than arbitrary commitments.

Eager verification simplifies later proof obligations because the fault proof or validity proof system may only need to show that derivation consumed the verified inbox entry. The tradeoff is cost: every DA commitment must be verified on L1 at submission time. Eager verification also depends on the shape of the inbox. If the inbox is an externally owned account or another mechanism that cannot execute validation logic, this pattern may not apply directly.

This is the easiest integration approach, and has been used by both EigenDA and Celestia in the beginning. But both have replaced this approach in later stages.
But this approach does not work for OP stack given that their inbox is an EOA, and not a smart contract that can execute logic.







\section{Attacks from Requirement Violations}
\label{sec:attacks}

This section derives concrete attacks from violations of the validation requirements. The examples are provider agnostic. They arise whenever settlement accepts an L2 state transition without enforcing the full relation among inbox input, DA commitment, external data, and rollup payload.

\subsection{Under-Constrained Data Acceptance}
\label{sec:attacks:binding}

The most direct integration failure is accepting an L2 state transition without
enforcing that the data used for derivation or execution is bound to the
accepted AltDA commitment. 

Suppose the L1 inbox contains bytes corresponding to a valid commitment \(c\). A DA verifier establishes that \(c\) is valid, but the settlement relation does not require the claimant to show \(bind(c,b)\) for the blob \(b\) actually consumed by the L2. A malicious party can then derive or prove a state transition using different data. If settlement accepts the resulting state transition, the accepted L2 state is no longer constrained by the data authenticated by the DA verifier.

This failure does not require breaking the cryptographic commitment scheme. The scheme may be binding, yet settlement may never check the binding relation for the blob used by execution. The result is an underconstrained settlement predicate: a DA commitment exists, but it is not connected to the state transition being accepted.

This root cause can lead to several concrete outcomes. First, an adversary may
post an invalid state assertion that cannot be defeated under the settlement
rules. Honest challengers may still raise a challenge, but if the dispute game
or validity verifier does not check the commitment-to-blob relation, the
malicious assertion may win.

Second, the adversary may create conflicting state claims from the same apparent
DA commitment. For example, two different blobs may lead to two different L2
states while both appear compatible with the same accepted inbox input. If the
bridge or settlement protocol does not enforce a unique data relation, the
adversary may be able to use inconsistent accepted states to withdraw funds or
otherwise exploit bridge accounting. The specific mechanics depend on the
settlement protocol, but the underlying violation is the same: the accepted
commitment does not uniquely constrain the L2 state transition.

\subsection{Undefined Failure Handling}

A second class of attacks arises when   invalid inputs are not handled as normal
adversarial inputs. A malicious batcher may post arbitrary bytes to the L1
inbox: an object that is too short, missing required fields, incorrectly
encoded, or not containing an AltDA commitment at all.  It may also post a well formed commitment with invalid DA attestation.

All such inputs should be rejected deterministically. If parsing or verification
failure instead causes the derivation software, proof program, or verifier to
panic, abort, or enter implementation specific behavior, the adversary can turn
invalid data into a liveness or safety attack.

In an optimistic derivation based system, this can become a bridge attack. The
malicious batcher first posts a malformed or invalid commitment that causes
honest derivation software to halt. The malicious proposer then posts an invalid
L2 state root to L1. Since honest nodes cannot continue deriving the correct L2
chain, honest challengers may be unable to compute or submit the challenge
needed before the challenge window expires. The invalid state root may then
finalize.

The defense is total failure semantics. The input domain of the parser is not the set of well formed commitments. It is the set of all byte strings that can appear in the L1 inbox. Parsing or verification failure must return \(Drop\), not cause a consensus level halt.

\subsection{Valid Commitment with Undecodable Data}

A final failure mode occurs after DA verification and binding succeed. A
malicious batcher may submit a valid and recent commitment to a blob that is
available and correctly bound, but whose contents cannot be interpreted as a
canonical rollup payload.

This case is distinct from blob unavailability. If the blob cannot currently be
retrieved, the appropriate result may be \(\mathsf{Stall}\). If the blob is
retrieved and bound but fails the rollup's payload relation, the input should be
classified as invalid and return \(\mathsf{Drop}\). Examples include missing
length fields, inconsistent padding, non-canonical encodings, or data that would
cause different implementations to produce different payloads.

If the codec panics, allocates unbounded memory, or produces
implementation specific output, honest nodes may halt or diverge. The defense
is the same totality principle applied to the blob-to-payload boundary: an
available and bound blob must either determine a unique
\(\mathsf{RollupPayload}\), or it must be rejected deterministically.

\begin{table}[t]
\centering
\caption{Validation obligations and failure modes.}
\label{tab:obligation-attacks}
\begin{tabular}{p{0.28\linewidth}p{0.62\linewidth}}
\toprule
Missing obligation & Representative consequence \\
\midrule
Canonical parsing &
Adversarial inbox bytes may be interpreted inconsistently, or malformed inputs
may halt derivation instead of being rejected. \\

DA verification &
An invalid or unauthenticated commitment may support an accepted state
transition. \\

Recency &
A historically valid but expired commitment may be consumed after honest
parties can no longer retrieve or dispute the corresponding data. \\

Commitment-to-blob binding &
The payload used by derivation may not be constrained by the commitment posted
to L1. \\

Payload decoding uniqueness &
Different honest implementations may decode the same blob into different
rollup inputs. \\

Deterministic failure handling &
Invalid, stale, unavailable, or undecodable inputs may lead to
implementation specific behavior rather than canonical \(Drop\) or \(Stall\)
outcomes. \\
\bottomrule
\end{tabular}
\end{table}

Table~\ref{tab:obligation-attacks} summarizes how each validation obligation
corresponds to a representative failure mode.

\section{Evaluation of Representative Integrations}
\label{sec:real-world}

This section applies the framework to representative AltDA integration designs. The purpose is not to rank DA providers. The purpose is to identify which semantic obligations are enforced by which components, and which obligations remain integration responsibilities.

\subsubsection{Methodology.}
We evaluate public specifications, repositories, and audit or documentation
artifacts for each integration. For each obligation in
Algorithm~\ref{alg:validation}, we ask whether the examined materials identify
an enforcement point in the inbox, DA verifier, derivation pipeline, fault proof
VM, or validity proof. We use ``not identified'' to mean that we did not find
explicit enforcement in the examined materials.

\begin{table*}[t]
\centering
\caption{Applying the validation framework to representative AltDA integrations. A check mark means the examined design appears to enforce the obligation in the relevant path; ``not identified'' means we did not find explicit enforcement in the examined materials.}
\label{tab:evaluation}
\begin{tabular}{p{2.8cm}p{2.4cm}p{2.1cm}p{2.2cm}p{2.0cm}p{2.7cm}}
\toprule
Integration & DA verifier Integration Pattern & DA verifier & Recency & Binding & E2E L2s type \\
\midrule
Celestia OP / Hana\cite{hana_code} & Validity proof & Blobstream & Not identified & Not identified &  Validity L2 \\
\midrule
Celestia Arbitrum\cite{celestia_arb} & Interactive fault proof & BlobStream & Not identified & Verified by merkle opening & Optimistic L2 \\
\midrule
EigenDA OP / Hokulea \cite{hokulea_repo}& Validity proof & Canoe (by sp1cc, steel) & Part of derivation check & Verified by kzg proof  & Validity and Hybrid L2 \\
\midrule
EigenDA Arbitrum \cite{eigeda_arb}& Eager inbox verification & DA verifier & Part of inbox verification & Verified by kzg proof & Optimistic L2 \\
\midrule
Avail / ZKsync \cite{avail_DA} & Eager inbox verification & Vector & Not identified & Verified by KZG proof & Validity L2 \\
\bottomrule
\end{tabular}
\end{table*}

\subsection{Blobstream Based Celestia Integrations}

Blobstream authenticates Celestia side commitments on Ethereum. At a high level,
Celestia validators attest to Celestia data roots, and Blobstream makes those
authenticated roots available to Ethereum. An L2 can then use Blobstream to
verify positive statements about Celestia data, such as inclusion of a
namespace, share range, or blob commitment under an authenticated Celestia root.

In our framework, Blobstream contributes to \(validDA\). It helps establish that some Celestia data is included under an authenticated root. However, a positive verification interface is not by itself a total validation function over all possible L2 inbox references. If an inclusion proof fails, that failure does not necessarily prove that the data was never stored by Celestia; it may only show that the submitted proof was invalid. Therefore, the Celestia integration must devise scheme to reject invalid AltDACommitment.

For an OP Celestia integration using Blobstream, the key question is whether the settlement path proves that the data consumed by OP derivation is bound to the Celestia commitment accepted from the L1 inbox.
At the examined Hana audit snapshot, we found DA authentication but not
end-to-end enforcement of the binding relation between the authenticated
commitment and the data consumed by OP derivation~\cite{hana_audit, hana_code}.
This is not merely an implementation detail: unless the final production proof
relation enforces this predicate, settlement remains underconstrained described in Section~\ref{sec:attacks:binding}.
This illustrates the type of issue that high level DA classifications can miss: a DA verifier may exist, and a DA commitment may be authenticated, while the rollup still fails to connect the authenticated commitment to the exact data consumed by derivation.

Celestia's Arbitrum integration illustrates a different composition~\cite{celestia_arb}. Arbitrum's interactive fault proof setting can combine DA verification with a bytes-opening check inside the dispute process. In this design, Blobstream need not act as a total non-membership oracle over all possible inbox references. The fault proof integration can instead enforce the relevant DA predicate at the disputed step. Under our framework, this is a valid way to satisfy the semantic obligations, provided the dispute relation also defines recency and failure handling.

In the Celestia integrations we examined, we did not identify an explicit recency rule tied to the rollup's dispute or proof window. This does not by itself prove exploitability; it identifies a condition that should be specified and audited.

\subsection{EigenDA Certificate Based Integrations}

EigenDA differs from bridge based designs in that its stake table and verification logic are Ethereum native. An EigenDA certificate can be verified by an Ethereum contract against stake information available on Ethereum, so the DA verifier does not need a separate bridge from another consensus system.

In EigenDA Hokulea OP integration, DA verification is enforced through a validity proof path~\cite{hokulea_repo}, see Section~\ref{sec:e2d_validity_proof}. The prover uses libraries such as SP1 contract call and Steel to prove execution against Ethereum contract state~\cite{sp1_cc,steel}. The validity proof library for EiegnDA DA verifier is called Canoe\cite{eigenda_canoe}. The proof system can show that the DA verifier accepted the EigenDA certificate and that derivation and execution followed the rollup rules. Under our framework, this places the DA predicate inside the validity relation. The security requirement is that the proof statement includes not only execution correctness but also the binding between the accepted EigenDA certificate and the payload consumed by derivation.

EigenDA's Arbitrum integration follows an eager verification pattern. The inbox contract checks the AltDA commitment at submission time, and failed commitments are not accepted into the derivation pipeline. This shifts DA verification to L1 execution and simplifies later settlement obligations, at the cost of verifying each commitment onchain.

In the EigenDA integrations we examined, we did not identify a missing check under the proposed framework. This statement is not a proof of complete security; it means that the integration pattern naturally aligns with the validation obligations identified in this paper.

\subsection{Avail and ZKsync Integration}

Avail has its own consensus protocol and uses SP1-Vector\cite{sp1_vector} to make Avail commitments available to Ethereum. At a high level, Avail produces commitments to published data, while Vector allows Ethereum contracts or proof systems to authenticate statements about Avail side data.

The Avail-ZKsync integration we examined follows an eager verification pattern: the inbox contract calls the Vector contracts to validate that the relevant data commitment is authenticated before treating it as rollup input. In our framework, this contributes to \(validDA\) at the inbox boundary. However, we did not identify an explicit recency condition tied to the rollup's proof or settlement window. As with the Celestia case, this should be read as an integration question to specify and audit, not as a complete exploit claim.

\section{Relation to DA Risk Classifications}
\label{sec:discussion}

Ecosystem frameworks such as L2Beat's data availability classifications identify
provider and DA verifier level DA assumptions: where data is posted, whether a DA
bridge exists, and what trust assumptions the DA provider introduces. Our
framework is complementary. It analyzes the integration semantics layer: whether
the rollup settlement path enforces a deterministic relation from adversarial
L1 inbox input to a unique rollup payload or a canonical failure outcome.

A DA verifier may establish \(validDA(c)\), but the surrounding L2 integration
must still enforce canonical parsing, recency, commitment-to-blob binding,
payload decoding, and deterministic failure handling. Thus, two L2s using
the same DA provider and bridge may have different security properties
depending on whether these obligations are enforced in the inbox contract,
derivation pipeline, fault proof VM, or validity proof.

\section{Limitations}

This paper focuses on the integration boundary between L2s and external DA systems. It does not analyze the internal security of any DA provider, such as committee honesty, erasure coding soundness, data availability sampling, or consensus finality. We assume the cryptographic and consensus assumptions of the underlying DA system hold.

The framework also does not prescribe a single correct engineering architecture. Eager inbox verification, interactive fault proof integration, and validity proof integration can all enforce the required semantics. The correct choice depends on cost, proof system constraints, inbox design, and the availability of authenticated settlement layer state.

Finally, our real world analysis is based on examined public specifications and code paths. A missing check in the examined path should be read as an audit target unless paired with a complete exploit demonstration. The goal is to identify the semantic obligations that integrations must satisfy and to make those obligations explicit enough to be reviewed.

\section{Conclusion}
\label{sec:conclusion}

AltDA systems allow Ethereum L2s to scale beyond native Ethereum data availability by moving bulk data publication outside Ethereum. This creates a new consensus critical boundary. A DA verifier alone is not sufficient: it may authenticate a fact about an external DA system without ensuring that the L2 state accepted by settlement is derived from the same data.

This paper presented an end-to-end validation framework for secure AltDA integration. We modeled the boundary as a typed, deterministic, and total translation from L1 inbox bytes to either a unique rollup payload or an explicit failure outcome. The model identifies obligations: inbox bytes parsing, DA verification, recency, commitment-to-blob binding, deterministic payload interpretation, and total failure handling. We showed how missing obligations lead to concrete attacks and used the framework to analyze representative integration patterns.

The central lesson is that AltDA security is not a provider level property and not a single verifier call. It is a property of the complete derivation, proof, and settlement path. As AltDA L2s secure increasing value, this boundary should be specified and audited with the same rigor as the rollup execution and settlement systems themselves.

\bibliography{references.bib}

\end{document}